\documentclass[12pt,leqno]{amsart}

\usepackage{amsthm}
\usepackage{amsmath}
\usepackage{amsfonts}

\newtheorem{theorem}{{\sc Theorem}}[section]

\newtheorem{cor}[theorem]{{\sc Corollary}}
\newtheorem{lemma}[theorem]{{\sc Lemma}}
\newtheorem{prop}[theorem]{{\sc Proposition}}
\newtheorem{defin}[theorem]{{\sc Definition}}

\theoremstyle{remark}
\newtheorem{remark}[theorem]{{\sc Remark}}

\def\f{\mathfrak }
\def\b{\mathbb }

\def\mm{\b M}
\def\nn{\b N}

\def\erw{\b E}

\def\phi{\varphi }

\def\calc{{\mathcal C}}

\def\calp{{\mathcal P}}

\def\frg{{\f g}}

\def\frk{{\f k}}

\def\frp{{\f p}}

\def\on{\operatorname}

\def\tra{^{\prime}}
\def\inv{^{-1}}

\def\dpst{\displaystyle}

\begin{document}

\title[Fluctuations of Wigner matrices]{Fluctuations of Wigner-type random matrices associated with symmetric spaces of class DIII and CI.}

\author{Michael Stolz}\thanks{Westf\"alische Wilhelms-Universit\"at M\"unster, Institute of Mathematical Stochastics, Orl\'eans-Ring 10, 48149 M\"unster, Germany. E-mail: michael.stolz@uni-muenster.de. Research supported by Deutsche Forschungsgemeinschaft (DFG) via SFB 878.}

\begin{abstract}
Wigner-type randomizations of the tangent spaces of classical symmetric spaces can be thought of as ordinary Wigner matrices on which additional symmetries have been imposed. In particular, they fall within the scope of a framework, due to Schenker and Schulz-Baldes, for the study of fluctuations of Wigner matrices with additional dependencies among their entries. In this contribution, we complement the results of these authors in that we develop a calculus of patterns which makes it possible to control the asymptotic contributions of dihedral non-crossing pair partitions for the Cartan classes DIII and CI, thus obtaining explicit CLTs for these cases.

\end{abstract}

\maketitle

\section{Introduction}

Much of random matrix theory is 
concerned with probability measures on the spaces of hermitian,
real symmetric, and quaternion real matrices. One of the reasons for this focus is the fact, proved by Freeman Dyson (\cite{Dyson}),
that any hermitian matrix (thought of as a
truncated Hamiltonian of a quantum system) that commutes with a
group of unitary symmetries and ``time reversals'' breaks down to
these three constituents, referred to by Dyson as the ``threefold way'', which are, in structural
terms, the tangent spaces to the Riemannian Symmetric Spaces (RSS)
of class A, AI and AII. But by the 1990s, the full ``tenfold way'' of infinite series of irreducible 
RSS had found its way into physics models, in particular in condensed matter theory. The structural reasons are explained in \cite{AS, AZ, HHZ, ZirnHbk, RSFL, AK}. Some basic aspects are summarized in \cite{azldp}.\\

Wigner's famous
result of 1958 (\cite{Wigner58}) states that for a symmetric
matrix with independent entries on and above the diagonal, the
mean empirical spectral distribution converges weakly to the
semicircle distribution as matrix size tends to infinity.
It is the prototype of a universality result in that it only depends on certain
assumptions about the moments of the matrix entries, but not on
the specifics of their distributions, and has been extended to the full tenfold way by Hofmann-Credner and the author in \cite{ECP}. It turns out that the semicircle distribution remains the limit law of the empirical spectral measures for seven of the ten families, while for the ``chiral'' classes with the Lie-theoretic labels AII, BD1, CII, it has to be replaced by a suitably transformed Mar\v{c}enko-Pastur distribution.\\

Wigner's theorem has been complemented by results about the corresponding fluctuations, initially under the assumption of Gaussianity of the matrix entries (\cite{Joh}). A level of generality comparable to Wigner's set-up was reached by Kusalik, Mingo, and Speicher in \cite{KMS}, even though that paper treats sample covariance rather than Wigner type matrices. The Wigner case is contained in the article \cite{SSB2} by Schenker and Schulz-Baldes. Actually, their work does much more, in that it substantially weakens the independence assumption on the upper-diagonal entries of standard Wigner matrices, to the effect that, e.g., additional symmetries can be enforced. A full statement of their main result will be given below. In the case of sample covariance matrices, an analogous weakening of independence assumptions was achieved by Friesen, L\"owe, and the author in  \cite{JMVA}.\\

While \cite{SSB2, JMVA} contain results on a high level of generality, additional work is required if one aims at explicit Central Limit Theorems for the fluctuations in specific models with explicitly given, non-trivial symmetries. The present contribution sets out to develop some convenient bookkeeping tools which are mighty enough to guide the explicit evaluation of formulae for the asymptotic variances of Wigner-type matrices associated with the symmetric spaces of Class DIII and CI. These classes are neither classical Wigner-Dyson nor chiral, and their role in the modeling of mesoscopic normal-superconducting hybrid structures has been pointed out by Altland and Zirnbauer in \cite{AZ}.\\

Section \ref{objects} below contains the definitions of the matrix ensembles of interest, as well as the statement of our Gaussian limit theorem. Section \ref{implementing} provides as much background on the framework set up by Schenker and Schulz-Baldes as is necessary to make the present work independent of any previous acquaintance with the details of \cite{SSB2}. Section \ref{d3} treats the case of Class DIII and, along the way, develops the core material of our approach. Section \ref{c1} adapts that material to the slightly different circumstances encountered in the case of Class CI.

\section{Matrix ensembles and results}
\label{objects} 

If $\frg$ is the Lie algebra, i.e., the tangent space (at any point) of a compact Lie group $G$, then $i \frg = \sqrt{-1}\, \frg$ is a space of hermitian matrices. If $\frg = \frk \oplus \frp$ is its decomposition into the $(+1)$-eigenspace $\frk$ and the $(-1)$-eigenspace $\frp$ of a (Cartan) involution, and $\frk$ is the Lie algebra of a compact Lie group $K$, then $\frp$ can be viewed as the 
tangent space of the Riemannian Symmetric Space $G/K$, and $i \frp$ as a convenient proxy if one prefers doing random matrix theory for matrices with real eigenvalues. In a nutshell, this is the rationale for attaching Lie-theoretic labels to the following two spaces of hermitian matrices.\\

\begin{description}

\item[Class DIII] $$ \mm_n^{{\rm DIII}} = \left\{ \left(
\begin{array}{rr} X_1 & X_2 \\ X_2 & -X_1
\end{array} \right):\ X_i \in (i {\mathbb R})^{n \times n}\ \text{\rm skew symmetric} \right\}$$

\item[Class CI] $$  \mm_n^{{\rm CI}} = \left\{ \left(
\begin{array}{rr} X_1 & X_2 \\ X_2 & -X_1
\end{array} \right):\ X_i \in {\mathbb R}^{n \times n}\ \text{\rm symmetric} \right\}$$
\end{description}

These spaces can be turned into Wigner type ensembles of random matrices as follows. Each matrix space induces a finite partition $\{B_i: i \in I\}$ (uniquely determined by the requirement that it possess minimal number of blocks) of the set of index pairs in $\{1, \ldots, 2n\}^{\times 2}$ 
which do not correspond to the diagonals of skew-symmetric blocks, such that for all $i \in I$ the following property holds:  As soon as the matrix entry which corresponds to an index pair from $B_i$ is determined, the matrix entries corresponding to all index pairs from $B_i$ are determined as well.  
Let $\{(p_i, q_i): i \in I\}$ be a system of representatives for the blocks, and let $(a_n(p_i, q_i))_{i \in I}$ be a family of independent centered complex random variables with the following properties:
\begin{itemize}
\item For each $k \in \nn$ the $k$-th absolute moments $\erw|a_n(p_i, q_i)|^k$ are uniformly bounded in $n$ and $i$.
\item There exists $\sigma^2 > 0$ such that $\erw|a_n(p_i, q_i)|^2 = \sigma^2$ for all $i \in I$.
\end{itemize}
Note that we have defined $I$ in such a way that the latter condition does not lead to conflict with the fact that the diagonal elements of skew-symmetric matrices are zero. With the 
$(a_n(p_i, q_i))_{i \in I}$ at hand, for $(r, s) \in B_i$ we define the matrix entry $a_n(r, s)$ as an identical copy (or possibly the negative of an identical copy) of $a_n(r, s)$, according to which algebraic relations among entries give rise to the class $B_i$. If $(r, s)$ corresponds to the diagonal of a skew symmetric block, we set $a_n(r, s) = 0.$\\

With these definitions in place, we define the matrix ensembles of interest as
$$ X_n^{\calc} =  \frac{1}{\sqrt{2n}} (a_n(p, q))_{p, q = 1, \ldots, 2n},$$
where $\calc \in \{\on{DIII}, \on{CI}\}.$ If $\calc$ is clear from the context, we drop the superscript.\\

To state our results, denote by $\on{T}_m$ the $m$-th Chebychev polynomial of the first kind, i.e., the one which satisfies the identity $\on{T}_m(2 \cos(\theta)) = 2 \cos(m\theta)$. Write $\on{T}_m(\cdot, \sigma)$ for the re-scaled version given by
$$ \on{T}_m(x, \sigma) = \sigma^m \on{T}_m\left( \frac{x}{\sigma}\right).$$

\begin{theorem}
\label{main}
Let $\calc \in \{ \on{DIII}, \on{CI}\}$. 
For each $M \in \nn$, the random vector
$$ \left(
(\on{Tr}(\on{T}_1(X_n^{\calc}, \sigma)) - \erw(\on{Tr}(\on{T}_1(X_n^{\calc}, \sigma)))), \ldots, 
(\on{Tr}(\on{T}_M(X_n^{\calc}, \sigma)) - \erw(\on{Tr}(\on{T}_M(X_n^{\calc}, \sigma))))
\right)$$
converges, as $n \to \infty$, to a centered Gaussian vector with covariance matrix 
$$\on{diag}(V^{\calc}(1), \ldots, V^{\calc}(M)),$$
where

$$ V^{\calc}(m) = \begin{cases} 
0, & m = 1,\\
\ast, & m=2,\\
0, & m \ge 3\ \text{\rm odd},\\
4 m \sigma^{2m}, & m \ge 4\ \text{\rm even}.
\end{cases}
$$

\end{theorem}

\begin{remark}
In the statement of Theorem \ref{main}, the value of $V^{\calc}(2)$ is unspecified, since it depends on the fourth moments of the $a_n(p, q)$, about which, apart from uniform boundedness, no assumptions are made in the present paper. Understanding the special role of $m=2$ requires substantially more background on the details of the proofs in \cite{SSB2} than what will be provided in Section \ref{implementing} below. In a nutshell, the reason is that for small $m$ the reduction to dihedral noncrossing pair partitions is subject to a few caveats.
\end{remark}

\section{Implementing symmetries: A review of SSB theory}
\label{implementing}

This section reviews some background on standard Wigner matrices (as opposed to the Wigner-type matrices which were introduced in Section \ref{objects}), where, for the purposes of this paper, a standard $n \times n$ Wigner matrix will be an hermitian matrix 
$$ X_n =  \frac{1}{\sqrt n} (a_n(p, q))_{p, q = 1, \ldots, n}$$ such that the family $a_n(p, q), p \le q$ of matrix entries on and above the diagonal consists of independent, but not necessarily identically distributed, random variables. We will assume that the $a_n(p, q)$ are centered, and for each $2 \le k \in \nn$ we require that the $k$-th moments $\erw|a_n(p, q)|^k$
be finite and uniformly bounded in $n, p, q$. It is customary to assume that the variance of the entries be equal to $\sigma^2 > 0$ for all $n, p, q$, and we will make this part of our notion of a standard Wigner matrix.\\

Our task is now to implement additional symmetries in standard Wigner matrices. A useful tool to this end is a formalism which has been introduced by Schenker and Schulz-Baldes (henceforth SSB) in \cite{SSB}. Writing $[n]$ for the set $\{ 1, 2, \ldots, n\}\ (n \in \nn)$, the idea is as follows: Given an equivalence relation $\sim_n$ on $[n]^2$, one stipulates that matrix entries $a_n(p_1, q_1), \ldots, a_n(p_{\nu}, q_{\nu})$ be independent whenever $(p_1, q_1), \ldots, (p_{\nu}, q_{\nu})$ are elements of $\nu$ different equivalence classes of $\sim_n$, whereas the joint distribution of a family $a_n(P)$, where $P$ runs through an equivalence class of $\sim_n$, is arbitrary. The interpretation we have in mind is that $\sim_n$ captures the symmetries that define the different matrix spaces from Section \ref{objects}, and that matrix entry random variables that correspond to $\sim_n$-equivalent index pairs are, up to a sign, identical (as measurable maps). But note that the scope of the SSB formalism is much broader.\\

The fluctuation results of SSB in \cite{SSB2}, which we are going to use, assume that certain quantitative characteristics of the $\sim_n$-equivalence classes do not grow too fast as a function of $n$. Specifically, it is assumed that the quantity 
\begin{equation}
 \label{alpha2} \alpha_2(n) = \max_{p, q \in [n]^2} \# \{ (r, s) \in [n]^2:\ (p, q) \sim_n (r, s)\}
\end{equation}
 is of order $\on{O}(n^{\epsilon})$ for all $\epsilon > 0$, and that $\hat{\alpha}_0(n) \alpha_2(n)^{\eta} = \on{o}(n^2)$ for all $\eta > 0$, where
\begin{equation} 
\label{alpha0}
\hat{\alpha}_0(n) = \# \{ (p, q, r) \in [n]^3:\ (p, q) \sim_n (q, r)\ \text{\rm and}\ p \neq r\}.
\end{equation}
Inspection of the list of spaces in Section \ref{objects} yields that for the cases of interest, $\alpha_2(n) \le 4$ and 
$\hat{\alpha}_0(n) = 0.$ So these conditions are always satisfied in our applications.\\

For investigating the fluctuations of the empirical spectral measure
of a standard Wigner matrix about its limit, i.e., the semicircle distribution, a well-established proxy are the fluctuations of  vectors of traces of nonnegative integer powers $\on{Tr}(X^k_i)$ ($i$ running through a finite set $I$). In \cite[Thm.\ 2.1]{SSB2} that we are going to treat as a black box, it is shown that the condition on the growth of $\alpha_2(n)$ alone suffices to guarantee that the joint cumulants of order $\ge 3$ of a family of 
$\on{Tr}(X^k_i)$ vanish as $n \to \infty$. This already yields a Gaussian limit theorem for the fluctuations, albeit one with possibly degenerate asymptotic variances.\\

The aspect of \cite{SSB2} which will be our starting point for explicit computiations, and which uses the growth condisions on both $\alpha_2$ and $\hat{\alpha}_0$, is concerned with expressions for the asymptotic cumulants of order $2$, i.e., the asymptotic covariances. Before stating what \cite{SSB2} has to say about these objects on the level of generality it aims at, let us digress a little and explain what computing covariances (or, for that matter, higher order cumulants) of traces of powers amounts to. As a matter of fact,
$$ \on{Tr}(X_n^k) = \frac{1}{n^{k/2}} \sum a_n(p_1, q_1) a_n(p_2, q_2) \ldots a_n(p_k, q_k),$$
where the sum is over all sequences of index pairs $(p_l, q_l) \in [n]^2$ that satisfy the {\it consistency relations}
\begin{equation}
\label{konsistenz}
 p_2 = q_1, p_3 = q_2, \ldots, p_k = q_{k-1}, p_1 = q_k.\end{equation}
Expanding bilinearly, this yields
\begin{equation}
\label{masterexpansion}
\on{Cov}\left( \on{Tr}(X_n^{k_1}), \on{Tr}(X_n^{k_2})\right)
= \frac{1}{n^{k/2}} \sum \on{Cov}\left( \prod_{l=1}^{k_1} a_n(P_{1, l}), \prod_{l=1}^{k_2} a_n(P_{2, l})\right),
\end{equation}
where the sum is over multi-indices $\mathbf{P} = (P_{i, l})_{i = 1, 2, l = 1, \ldots, k_l} = (p_{i, l}, q_{i, l})_{i = 1, 2, l = 1, \ldots, k_l}$ such that the $(p_{i, l}, q_{i, l})\ (i = 1, 2)$ satisfy consistency relations as in \eqref{konsistenz}.\\

If we follow \cite{SSB2} and denote by $\calp_{[k_1] \cup [k_2]}$ the set $\{ (i, l):\ i = 1, 2, l \in [k_i]\}$, then a crucial observation is as follows: The equivalence relation $\sim_n$ on $[n]^2$ induces an equivalence relation $\pi$ on $\calp_{[k_1] \cup [k_2]}$ by requiring that $(i, l)$ and $(i\tra, l\tra)$ belong to the same class of $\pi$ if, and only if, $ P_{i, l} \sim_n P_{i\tra, l\tra}$. This makes it possible to classify the multi-indices
$\mathbf{P} = (P_{i, l})_{i = 1, 2;\ l = 1, \ldots, k_l}$ according to which $\pi$ they induce on $\calp_{[k_1] \cup [k_2]}$. The technical core of the proofs in \cite{SSB2}, then, consists in arguments to the effect that certain types of $\pi$ are always associated to summands which, if multiplied by the prefactor $\frac{1}{n^{k/2}}$, give a negligible contribution in the $n \to \infty$ limit. It turns out that the summands in \eqref{masterexpansion} which actually contribute to the limit are all associated to partitions of $\calp_{[k_1] \cup [k_2]}$ of a very restrictive type. Observe that a partition of $\calp_{[k_1] \cup [k_2]}$ typically consists both of blocks which are subsets of $\{1\} \times [k_1]$ or of $\{2\} \times [k_2]$, and of blocks which contain elements of both sets. If we think of the elements of $\{1\} \times [k_1]$ as marked points on the inner boundary circle of an annulus in the complex plane, and of the elements of $\{2\} \times [k_2]$ as marked points on the outer circle, then the latter kind of blocks can be thought of as connections between the circles. A pictorial (and only approximately correct) description of the partitions which contribute nontrivially to the large $n$ limit is as follows: All blocks are pairs, and the lines that connect the circles do not cross.\\

Rather than make this description more precise, which would require  a host of extra terminology, we take advantage of the fact that the problem to describe asymptotically non-negligible partitions has a much neater solution if one replaces the covariance in \eqref{masterexpansion} with the covariance of traces of Chebyshev polynomials in matrix arguments.\\

Let $g: [m] \to [m]$ be a bijection which, as an element of the symmetric group $\on{Sym}([m])$, is contained in the subgroup $D$ generated by the cycle $\gamma = (1\, 2\, 3\, \ldots m)$ and the involution $\tau = (1\, m)(2, m-2)\ldots (\frac{m}{2}-1\, \frac{m}{2})$ if $m $ is even, or $\tau =
(1\, m)(2, m-2)\ldots (\frac{m+1}{2}-1, \frac{m+1}{2} + 1)$ if $m$ is odd. Note that $\gamma$ maps $m$ to $1$, $1$ to $2$, and so on, and that $\tau$ swaps $1$ with $m$, $2$ with $m-1$ and so on, and fixes $\frac{m+1}{2}$ if $m$ is odd. It is easily verified that $\tau\inv \gamma \tau = \tau \gamma \tau = \gamma\inv$, so $D$ is isomorphic to the dihedral group $\on{D}_{2m}$, see, e.g., \cite[Lemma 2.14]{Isaacs}. Then SSB define a partition 
$\hat{\pi}_g := \{ \{(1, l), (2, g(l)\}:\ l \in [m]\}$ and call it a dihedral non-crossing pair partition. Note that if we embed $[m]$ on the unit circle as the set of $m$-th roots of unity (with the natural labelling), then $\on{D}_{2m}$ can be thought of as the set of permutations that map neighboring points to neighboring points. This provides a link to the informal description of the asymptotically contributing partitions given above.\\

We are now in a position the state a result from \cite{SSB2} which we are going to apply subsequently.

\begin{prop}{\cite[Thm.\ 2.4]{SSB2}}
\label{asymptcov} 
~
$$
\on{Cov}(\on{Tr}(\on{T}_m(X_n,\sigma)),\on{Tr}(\on{T}_{\mu}(X_n,\sigma))) \ = \ \delta_{m\mu}  V_n(m)  +  \on{o}(1) ,
$$

with

\begin{equation}
 V_n(m) \ = \  \begin{cases} \frac{1}{n}
 \dpst \sum_{\substack{p,q \\ (p,p) \sim_n(q,q)}}
\erw(a_n(p,p) a_n(q,q)), & m = 1, \\
\frac{1}{n^2}
\dpst \sum_{\substack{p,q,r,s, \\ p \neq q\, , \ r \neq s, \\
(p,q) \sim_n (r,s)} } \on{Cov}( |a_n(p,q)|^2 ,|a_n(r,s)|^2 ), & m = 2,\\
\frac{1}{n^m} \dpst \sum_{g\in \on{D}_{2m}}
\dpst\sum_{\mathbf{P}\in S_n^{\on{good}}(\hat{\pi}_g)}
 \prod_{l=1}^m \erw(a_n(P_{1,l}) a_n(P_{2,g(l)})), & m \ge 3.
                                            \end{cases}
\end{equation}
Here $S_n^{\on{good}}(\hat{\pi}_g)$ denotes the set of multi-indices which induce the partition $\hat{\pi}_g$ on $\calp_{[m] \cup [m]}$ and do not meet a coordinate pair for which the corresponding matrix entry must be zero in view of the symmetries of the matrix spaces from Section \ref{objects} (i.e., the diagonal of a skew symmetric block in the DIII case).
\end{prop}

Note that in the statement of Prop.\ \ref{asymptcov} we have inserted our $S_n^{\on{good}}(\hat{\pi}_g)$ in the place of a slightly differently defined set of multi-indices in \cite{SSB2}. In fact, the multi-indices of \cite{SSB2} only stay clear of the main diagonal. Since the number of exceptions grows like $n$, and the total number of matrix entries grows like $n^2$, it is not very surprising that these modifications should do no harm. A more careful argument can be based on Remark 4.7 in \cite{SSB2}.

\newpage

\section{Class DIII}
\label{d3}
\subsection{Fundamentals}
\label{d3fundamentals}
The following basic lemma can be immediately read off from the definition of the DIII matrix space in Section \ref{objects}.

\begin{lemma}
\label{blocksd3}
Let $a, b \in [n], a \neq b$. Write
\begin{align*}
\calc_1(a, b) &:= \{ (a, b), (b, a), (n+a, n+b), (n+b), (n+a) \},\\
\calc_2(a, b) &:= \{ (n+a, b), (n+b, a), (a, n+b), (b, n+a)\},
\end{align*} and observe that $\calc_i(a, b) = \calc_i(b, a)\ (i = 1, 2).$
Then the equivalence classes which are induced on $[2n] \times [2n]$ by the symmetries of the tangent space of class DIII are precisely the $\calc_i(a, b)$ for $i = 1, 2, a, b \in [n], a \neq b$.
\end{lemma}

Now we are going to set up a language which will allow to apply the general theory that was reviewed in Section \ref{implementing} to the equivalence relation from Lemma \ref{blocksd3}. To this end, for $a, b \in [n], a \neq b,$ denote by $A(a, b)$ the ``aligned'' matrix $\left( \begin{array}{ll} a & b\\ a & b 
\end{array}\right),$ and by $R(a, b)$ the ``reversed'' matrix 
$\left( \begin{array}{ll} a & b\\ b & a 
\end{array}\right).$
\begin{defin}
An ($m$-)\emph{pattern} is a sequence $(\Delta_l, \Lambda_l)_{l \in [m]}$, where $\Delta_l$ is one of the $2 \times 2$ matrices
\begin{align}
\label{d3-delta-gleiche}
 \left( \begin{array}{ll} 0 & 0\\ 0 & 0 
\end{array}\right),  \left( \begin{array}{ll} 1 & 1\\ 1 & 1 
\end{array}\right),  \left( \begin{array}{ll} 0 & 1\\ 0 & 1 
\end{array}\right),  \left( \begin{array}{ll} 1 & 0\\ 1 & 0 
\end{array}\right),\\
\label{d3-delta-verschiedene}
  \left( \begin{array}{ll} 0 & 0\\ 1 & 1 
\end{array}\right),  \left( \begin{array}{ll} 1 & 1\\ 0 & 0 
\end{array}\right),  \left( \begin{array}{ll} 1 & 0\\ 0 & 1 
\end{array}\right),  \left( \begin{array}{ll} 0 & 1\\ 1 & 0 
\end{array}\right).
\end{align}
and $\Lambda_l \in \{ A(\cdot, \cdot), R(\cdot, \cdot)\}.$
For any choice $(a_l)_{l \in [m]} \in [n]^m, (b_l)_{l \in [m]} \in [n]^m$ such that $a_l \neq b_l$ for all $l \in [m]$, we call the sequence $(n \Delta_l + \Lambda_l(a_l, b_l))_{l \in [m]}$ an \emph{instance} of the pattern $(\Delta_l, \Lambda_l)_{l \in [m]}$. 
\end{defin}

Note that \eqref{d3-delta-gleiche} and \eqref{d3-delta-verschiedene} contain precisely those binary $2 \times 2$ matrices for which the sum over all entries is $\equiv 0 (\on{mod} 2)$.

\begin{defin}
Any of the matrices

$$  \left( \begin{array}{ll} 0 & 1\\ 0 & 1 
\end{array}\right),  \left( \begin{array}{ll} 1 & 0\\ 1 & 0 
\end{array}\right),
     \left( \begin{array}{ll} 1 & 0\\ 0 & 1 
\end{array}\right),  \left( \begin{array}{ll} 0 & 1\\ 1 & 0 
\end{array}\right)$$

is a \emph{step}. Any of the remaining matrices in \eqref{d3-delta-gleiche}, \eqref{d3-delta-verschiedene} is a \emph{plateau}. $\left( \begin{array}{ll} p_{1l} & q_{1l}\\ p_{2l} & q_{2l} 
\end{array}\right)$ is called a step if $\Delta_{l}$ is a step.\\

\end{defin}

\begin{defin}
~
\begin{itemize}
\item[(a)] Let $(\Delta_l, \Lambda_l)_{l \in [m]}$ be a pattern. 
If for each $l \in [m]$ we represent $n \Delta_l + \Lambda_l(a_l, b_l)$ as a matrix 
$\left( \begin{array}{ll} p_{1l} & q_{1l}\\ p_{2l} & q_{2l} 
\end{array}\right)$, we call $(n \Delta_l + \Lambda_l(a_l, b_l))_{l \in [m]}$ \emph{consistent} if both $(p_{1l}, q_{1l})_{l \in [m]}$ and $(p_{2l}, q_{2l})_{l \in [m]}$ are consistent in the sense of Section \ref{implementing}, i.e., $p_{i1} = q_{im}, p_{i2} = q_{i1}, \ldots, p_{im} = q_{i(m-1)}\ (i = 1, 2).$ 
\item[(b)] We call $(\Delta_l, \Lambda_l)_{l \in [m]}$ \emph{negligible}, if the number of consistent instances it admits is of order $\on{o}(n^m)$. Otherwise we call it \emph{substantial}.
\end{itemize}
\end{defin}

\begin{defin}
~
\begin{itemize}
\item[(a)] We say that $(\Delta_l)_{l \in [m]}$ satisfies the \emph{domino condition}, if for each $l \in [m]$ the second column of $\Delta_l$ equals the first column of $\Delta_{l+1}$ (where $m+1$ is identified with $1$). In this case, we also refer to $(\Delta_l)_{l \in [m]}$ as a domino-$\Delta$-sequence.
\item[(b)] We say that $(\Delta_l)_{l \in [m]}$ satisfies the \emph{reverse domino condition}, if for each $l \in [m]$, writing
\begin{equation}
\label{deltadarstellung} \Delta_l = \left( \begin{array}{ll} \alpha_l & \beta_l\\ \gamma_l & \delta_l 
\end{array}\right),\quad \Delta_{l+1} = \left( \begin{array}{ll} \alpha_{l+1} & \beta_{l+1}\\ \gamma_{l+1} & \delta_{l+1} 
\end{array}\right),
\end{equation}
 one has
$ \alpha_{l+1} = \beta_{l},\ \delta_{l+1} = \gamma_l$ (where, again, 
$m+1$ is identified with $1$). In this case, we also refer to 
$(\Delta_l)_{l \in [m]}$ as a reverse (domino) $\Delta$-sequence.
\end{itemize}
\end{defin}

\begin{lemma}
\label{substantial}
~
\begin{itemize}
\item[(a)] Suppose that the pattern $(\Delta_l, \Lambda_l)_{l \in [m]}$ is such that $(\Delta_l)_{l \in [m]}$ satisfies the domino condition. Then for it to be substantial it is necessary and sufficient that all $\Lambda_l\ (l \in [m])$ be equal to $A(\cdot, \cdot)$. 
\item[(b)] Suppose that the pattern $(\Delta_l, \Lambda_l)_{l \in [m]}$ is such that $(\Delta_l)_{l \in [m]}$ satisfies the reverse domino condition. Then for it to be substantial it is necessary and sufficient that all $\Lambda_l\ (l \in [m])$ be equal to $R(\cdot, \cdot)$. 
\end{itemize}
\end{lemma}
\begin{proof}
(a) Sufficiency is clear. For necessity, assume first that there exists $l \in [m]$ such that $\Lambda_l = A(\cdot, \cdot), \Lambda_{l+1} = R(\cdot, \cdot)$. We may assume without loss that $l = 1$. In the special case that $\Delta_1 = \Delta_2 = \left( \begin{array}{ll} 0 & 0\\ 0 & 0 
\end{array}\right)$, our assumptions imply that
$$ \left( \begin{array}{ll} p_{11} & q_{11}\\ p_{21} & q_{21} 
\end{array}\right) = \left( \begin{array}{ll} a & b\\ a & b 
\end{array}\right), \quad \left( \begin{array}{ll} p_{12} & q_{12}\\ p_{22} & q_{22} 
\end{array}\right) = \left( \begin{array}{ll} b & c\\ c & b 
\end{array}\right).$$
Consistency implies that $c = b$, so the number of consistent instances will be at most of order $\on{O}(n^{m-1})$. In the general case for $\Delta_1, \Delta_2$, this argument is still valid, since the pattern satisfies the domino condition, and since given $\Delta_1$ there are only two choices for $\Delta_2$. The remaining case of (a) ($\Lambda_l = R(\cdot, \cdot)$ for all $l$) as well as (b) can be handled along similar lines.
\end{proof}

In the two subsections that follow, this machinery will be applied to the evaluation of the various sums over multi-indices $\mathbf{P}$ from Proposition \ref{asymptcov} ($m \ge 3$ case). Note that Lemma \ref{substantial} permits to sum only over multi-indices that correspond to substantial patterns and confine the others to the $\on{o}(1)$ term. Also note that when we apply Proposition \ref{asymptcov} to our block matrix spaces, $2n$ will play the role of the $n$ of that proposition.

\subsection{Multi-indices compatible with a cyclic shift in $\on{D}_{2m}$.}
\label{d3cycle}

In this subsection we consider an element $g \in \on{D}_{2m}$ of the form $g = \gamma^{\nu}$, where $\nu \in \nn$ and $\gamma$ is the $m$-cycle $(1\, 2\, 3\, \ldots\, m)$ as in the paragraph before Prop.\ \ref{asymptcov}. For ${\mathbf P} \in S^{\on{good}}(\hat{\pi}_g)$ let $(\Delta_l({\mathbf P}), \Lambda_l({\mathbf P}))_{l \in [m]}$ be the pattern of which the sequence
$$ \left( \begin{array}{ll} p_{1, l} & q_{1, l}\\ p_{2, (l+\nu)} & q_{2, (l + \nu)} 
\end{array}\right)_{l \in [m]} $$ is an instance. By consistency of 
${\mathbf P}$, $(\Delta_l({\mathbf P}), \Lambda_l({\mathbf P}))_{l \in [m]}$ will satisfy the domino condition, and it is substantial if, and only if, all $\Lambda_l({\mathbf P})$ are equal to $A(\cdot, \cdot)$. Then
\begin{align}
\frac{1}{(2n)^m} \sum_{P \in S_n^{\on{good}}(\hat{\pi}_g)} &\prod_{l=1}^m \erw(a_n(P_{1, l}) a_n(P_{2, g(l)})) \nonumber\\ 
\label{leadingsubstantial}= 
\frac{1}{(2n)^m} \sum_{\substack{ P \in S_n^{\on{good}}(\hat{\pi}_g),
\\ (\Delta_l({\mathbf P}), \Lambda_l({\mathbf P}))\\ \text{\rm substantial.}
}} &\prod_{l=1}^m \erw(a_n(P_{1, l}) a_n(P_{2, g(l)})) + \on{o}(1).
\end{align}

Now we turn to the evaluation of the leading order term, for which we have to consider substantial patterns only. Observe that it is not the choices of $a_l, b_l$ but only the choice of $\Delta_l$ that determines the value of $\erw(a_n(p_{1,l}, q_{1,l}) a_n(p_{2,(l+\nu)}, q_{2, (l+\nu)}))$. Indeed, if $\Delta_l$ is one of the four possibilities in \eqref{d3-delta-gleiche}, then $(p_{1,l}, q_{1,l}) = (p_{2,(l+\nu)}, q_{2,(l+\nu)})$ and thus
$$ \erw(a_n(p_{1l}, q_{1l})a_n(p_{2,(l+\nu)}, q_{2,(l+\nu)}) =  \erw(a_n(p_{1,l}, q_{1,l})^2) = - \erw(a_n(p_{1,l}, q_{1,l}) \overline{a_n(p_{1,l}, q_{1,l})}) = - \sigma^2$$ in view of our variance requirement for matrix entries away from the block diagonals. On the other hand,
$ \Delta_l =  \left( \begin{array}{ll} 0 & 0\\ 1 & 1 
\end{array}\right)$ means $(p_{2,(l+\nu)}, q_{2,(l+\nu)}) = (n + p_{1,l}, n+q_{1,l})$, hence $a_n(p_{2, (l+\nu)}, q_{2,(l+\nu)}) = - a_n(p_{1,l}, q_{1,l}) = 
\overline{ a_n(p_{1,l}, q_{1,l})}$ and thus $\erw( a_n(p_{1,l}, q_{1,l}) a_n(p_{2,(l+\nu)}, q_{2,(l+\nu)})) = \sigma^2.$ Interchanging the roles of $(p_{1,l}, q_{1,l})$ and $(p_{2,(l+\nu)}, q_{2,(l+\nu)})$ leads to the same result for $ \Delta_l =  \left( \begin{array}{ll} 1 & 1\\ 0 & 0 
\end{array}\right)$. As to $\Delta_l = \left( \begin{array}{ll} 1 & 0\\ 0 & 1 
\end{array}\right)$, we then have that $(p_{1,l}, q_{1,l}) = (n+a, b), (p_{2,(l+\nu)}, q_{2,(l+\nu)}) = (a_l, n+b_l)$, hence $a_n(p_{1l}, q_{1l}) = a_n(p_{2l}, q_{2l})$ and thus $\erw( a_n(p_{1,l}, q_{1,l}) a_n(p_{2,(l+\nu)}, q_{2,(l+\nu)})) = - \sigma^2.$ Again switching roles, we obtain the same result for $\Delta_l = \left( \begin{array}{ll} 0 & 1\\ 1 & 0 
\end{array}\right).$ \\

Taking also the domino condition into account, these considerations yield:

\begin{lemma}
\label{gleichbleibtgleich}
~
\begin{itemize}
\item[(a)] If $(p_{1l}, q_{1l}) = (p_{2l}, q_{2l})$ for some $l \in [m]$, then the same holds true for all $l \in [m]$, and the product of the covariances equals $(-1)^m \sigma^{2m}$.
\item[(b)] If $(p_{1l}, q_{1l}) \neq (p_{2l}, q_{2l})$ for some $l \in [m]$, then the same holds true for all $l \in [m]$, and the product of the covariances equals $(-1)^{\# \{ \text{\rm steps\ in\ the\ multi-index}\}} \sigma^{2m}$. Now, by consistency, the number of steps is even, and so the product of the covariances equals $\sigma^{2m}$.
\end{itemize}
\end{lemma}
Lemma \ref{gleichbleibtgleich} implies that for $m$ odd, the leading order term in \eqref{leadingsubstantial} vanishes. So we only have to consider the case that $m = 2d\ (d \in \nn)$. To this end, it is fundamental to observe that the sequence of $2 \times 2$--matrices $\Delta_l$ and the sequence of scalars $a_l, b_l$ may be chosen independently of one another. Furthermore, the number of  $\Delta$-sequences satisfying the domino condition (domino-$\Delta$-sequences for short) with both rows equal is the same as the number of domino-$\Delta$-sequences  with both rows different. The number of domino-$\Delta$-sequences  with both rows equal is two times the number of domino-$\Delta$-sequences  matrices with both rows equal and $\alpha_1 = 1$, where we have represented $\Delta_1$ as in \eqref{deltadarstellung}. This implies that the asymptotic contribution is
\begin{equation}
\label{zaehlproblem-id}
 4 \sigma^{2m}\ \# \{ {\mathbf P}: {\mathbf P}\ \text{\rm has\ identical\ rows\ and}\ p_{11} = 1\}.
 \end{equation}
To determine this combinatorial quantity, observe that a plateau must be of the form $\left( \begin{array}{ll} 0 & 0\\ 0 & 0 
\end{array}\right)$ if the rightmost step that precedes it is of the form $\left( \begin{array}{ll} 1 & 0\\ 1 & 0 
\end{array}\right)$, and that 
a plateau must be of the form $\left( \begin{array}{ll} 1 & 1\\ 1 & 1 
\end{array}\right)$ if the rightmost step that precedes it is of the form $\left( \begin{array}{ll} 0 & 1\\ 0 & 1 
\end{array}\right)$. So the multi-indices in \eqref{zaehlproblem-id} are uniquely determined once the places at which the steps occur (if any) are fixed. Observe that it suffices to specify the set (necessarily of even order) at which the steps occur, since the steps occur in alternating order, beginning with a $\left( \begin{array}{ll} 1 & 0\\ 1 & 0 
\end{array}\right)$. So the combinatorial factor in \eqref{zaehlproblem-id} equals
$$ \sum_{j=0}^d {m \choose 2j} = 2^{(m-1)}.$$

Summarizing the discussion in this subsection, we have seen that for $g$ a cyclic shift, the leading order term in \eqref{leadingsubstantial} is
\begin{equation}
\label{summaryshift}
\begin{cases}
0, & m\ \text{\rm odd},\\
2^{(m+1)} \sigma^{2m}, & m\ \text{\rm even}.
\end{cases}
\end{equation}

\subsection{Multi-indices compatible with a product $g = \tau \gamma^{\nu}$}
\label{d3-tau}

Recall from the paragraph before Prop.\ \ref{asymptcov} that $\tau$ is the reflection that swaps $1$ with $m$, $2$ with $m-1$ and so on. Then we are facing a consistency requirement of the form
$$ \left( \begin{array}{ll} p_{1,l} & q_{1,l}\\ p_{2, \nu + l-1} &  
q_{2, \nu + l-1}
\end{array}\right) = \Delta_l + \Lambda_l(a, b)$$
$$ \left( \begin{array}{ll} p_{1,l+1} & q_{1,l+1}\\ p_{2, \nu + l-2} &  
q_{2, \nu + l-2}
\end{array}\right) = \Delta_{l+1}+ \Lambda_{l+1}(a, b)$$
for all $l \in [m]$. In the one hand, this implies that $(\Delta_l)_{l \in [m]}$ must satisfy the reverse domino condition. In view of Lemma \ref{substantial}, the pattern $(\Delta_l, \Lambda_l)_{l \in [m]}$ will be negligible unless for all $l$ we have that $\Lambda_l = R(\cdot, \cdot).$\\

We have seen in Subsection \ref{d3cycle}, that it is the induced sequence $(\Delta_l)_{l \in [m]}$ that determines the sign of the product of covariances that comes from a $\hat{\pi}_g$-compatible multi-index $\mathbf{P}$. But observe that replacing the condition 
$\Lambda_l = A(\cdot, \cdot)$ from Subsection \ref{d3cycle} by 
$\Lambda_l = R(\cdot, \cdot)$ means that the plus and minus signs are swapped with respect to the previous subsection. This means that now the matrices
$$ \left( \begin{array}{ll} 0 & 0\\ 0 & 0 
\end{array}\right), \left( \begin{array}{ll} 1 & 1\\ 1 & 1 
\end{array}\right), \left( \begin{array}{ll} 0 & 1\\ 0 & 1 
\end{array}\right), \left( \begin{array}{ll} 1 & 0\\ 1 & 0 
\end{array}\right), \left( \begin{array}{ll} 1 & 0\\ 0 & 1 
\end{array}\right), \left( \begin{array}{ll} 0 & 1\\ 1 & 0
\end{array}\right)$$
will contribute a factor of $+\sigma^2$, while the matrices
$$ \left( \begin{array}{ll} 1 & 1\\ 0 & 0 
\end{array}\right), \left( \begin{array}{ll} 0 & 0\\ 1 & 1 
\end{array}\right)$$
will contribute a factor of $- \sigma^2$. \\

Now we study the step from $\Delta_l$ to $\Delta_{l+1}$ according to the reverse domino condition. It turns out that the eight matrices that occur in $\Delta$-sequences fall into two disjoint subsets that are closed under these steps. They are given in \eqref{d3invar1} and \eqref{d3invar2}. 
\begin{equation}
\label{d3invar1}
\left( \begin{array}{ll} 1 & 1\\ 1 & 1 
\end{array}\right), \left( \begin{array}{ll} 1 & 0\\ 0 & 1 
\end{array}\right), \left( \begin{array}{ll} 0 & 0\\ 0 & 0 
\end{array}\right), \left( \begin{array}{ll} 0 & 1\\ 1 & 0 
\end{array}\right)
\end{equation}
Note that all matrices in this group correspond to positive covariances.
\begin{equation}
\label{d3invar2}
\left( \begin{array}{ll} 0 & 1\\ 0 & 1 
\end{array}\right), \left( \begin{array}{ll} 1 & 0\\ 1 & 0 
\end{array}\right), \left( \begin{array}{ll} 0 & 0\\ 1 & 1 
\end{array}\right), \left( \begin{array}{ll} 1 & 1\\ 0 & 0 
\end{array}\right)
\end{equation}
This group contains the two matrices that correspond to negative covariances.\\

Now we turn to the evaluation of the sum in \eqref{leadingsubstantial} in the case $g = \tau \gamma^{\nu}$. For ease of exposition we consider the special case $\nu = 0$, i.e., $g = \tau$. The case of a general $\nu$ will lead to the same result, since even as a $ \nu \in \{1, \ldots, m-1\}$ means that the lower row of the obvious matrix representation of ${\mathbf P}$ is cyclically shifted w.r.t.\ the upper row, the images of the blocks of the equivalence relation that corresponds to $g = \tau$ under the shift will be exactly the blocks of the equivalence relation that corresponds to $ g = \tau \gamma^{\nu}$.\\

First we reduce the problem to the study of $\Delta$-sequences. The sum of interest is over ${\mathbf P}$ with $(\Delta_l({\mathbf P}), \Lambda_l({\mathbf P}))$, meaning that $\Lambda_l({\mathbf P}) = R(\cdot, \cdot)$ for all $l \in [m]$.  
To illustrate what compatibility with $\tau$ means for the $\Lambda$-sequence, consider the case $m = 5$. In this case we obtain a structure of the form

$$ \left( \begin{array}{ll} a & b\\ a & e 
\end{array}\right), \left( \begin{array}{ll} b & c\\ e & d 
\end{array}\right), \left( \begin{array}{ll} c & d\\ d & c 
\end{array}\right), \left( \begin{array}{ll} d & e\\ c & b 
\end{array}\right), \left( \begin{array}{ll} e & a\\ b & a
\end{array}\right).
$$
The third matrix is $R(c, d)$. The upper row of the second matrix and the lower row of the fourth matrix together constitute $R(b, c)$, and so on. If $m$ is even, no matrix in the sequence is of the form $R(\cdot, \cdot)$, but the pairings of upper and lower rows as indicated by $\tau$ still lead to $R$-matrices. Consequently, if the pattern is substantial, compatibility with $\tau$ on the level of the sequences $(a\, b\, \ldots) \in [n]^m$ comes for free. \\

Now we turn to the possible sequences of $\Delta$s. Suppose that
$\Delta_1$ has been chosen. Recall that if $g = \tau$, the choice of $\Delta_1$ fixes the upper row of the first matrix an the lower row of the $m$-th matrix. Taking again the case $m = 5$ for illustration, the starred positions are fixed by this choice and by consistency: 

$$ \left( \begin{array}{ll} \ast & \ast\\ \ast & ~ 
\end{array}\right), \left( \begin{array}{ll} \ast & ~\\ ~ & ~ 
\end{array}\right), \left( \begin{array}{ll} ~ & ~\\ ~& ~ 
\end{array}\right), \left( \begin{array}{ll} ~ & ~\\ ~ & \ast 
\end{array}\right), \left( \begin{array}{ll} ~ & \ast\\ \ast & \ast
\end{array}\right).
$$

Now proceed in the upper row from left to right. In the example, the first slot without a star is the upper right slot of the second matrix. We may choose the corresponding entry arbitrarily from $\{0, 1\}$. The upper left entry of the third matrix is then determined by consistency, while the upper right entry can again be freely chosen from $\{ 0, 1\}$, and so on. The last entry that can be freely chosen in this way is the upper right entry of the fourth (i.e., the $(m-1)$th) matrix, since the upper right entry of the last matrix is already determined by consistency with the first matrix. In total, given $\Delta_1$, we have had $2^{(m-2)}$ choices for the first row.\\

We now claim that there is one, and only one, way to fill the empty slots in the lower row with elements from $\{0, 1\}$ such that we obtain a $\Delta$-sequence induced by a multi-index which is compatible with $\hat{\pi}_{\tau}$. Proceeding from right to left, the first empty slot is the lower left slot of the fourth (i.e., $(m-1)$-th) matrix. Now, the lower row of the $(m-1)$-th matrix is paired, under $\tau$, with the upper row of the second matrix. Taken together, both rows constitute the matrix $\Delta_2$. In particular, the sum of the entries of both rows must be $\equiv 0 (\on{mod} 2)$. This problem has a unique solution in $\{0, 1\}$. Choosing the lower left entry of the $(m-1)$-th matrix as this solution, the lower right entry of the $(m-2)$-th matrix is fixed by consistency. As above, the lower left entry of this matrix is determined by congruence, this time with respect to the upper row of the third matrix. In this way, one proceeds in a uniquely determined way from right to left, until the first matrix is reached, where the lower left entry is already fixed by consistency with the $m$-the matrix.\\

If $\Delta_1$ is from the list in \eqref{d3invar1}, all $\Delta_l$ constructed in this way will also belong to this list, and the corresponding ${\mathbf P}$ will contribute $+ \sigma^{2m}$ to \eqref{leadingsubstantial}. If 
$\Delta_1$ is from the list in \eqref{d3invar2}, all $\Delta_l$ will also belong to this list, and this time negative factors may occur, namely, factors related to the plateau matrices 
\begin{equation}
 \left( \begin{array}{ll} 0 & 0\\ 1 & 1 
\end{array}\right),\quad\left( \begin{array}{ll} 1 & 1\\ 0 & 0 
\end{array}\right). 
\end{equation}

\begin{lemma}
\label{tauoddeven}
If $m$ is odd, then any $\hat{\pi}_{\tau}$-consistent multi-index
such that one, and thus all, of the matrices from the associated $\Delta$-sequence belong to the list in \eqref{d3invar2}, will contribute $- \sigma^{2m}$ to the sum in \eqref{leadingsubstantial}. If $m$ is even, each such multi-index contributes $+ \sigma^{2m}$.
\end{lemma}
Since we have already remarked that any of the possible eight choices for $\Delta_1$ gives rise to $2^{m-2}$ $\Delta$-sequences such that the corresponding ${\mathbf P}$ are compatible with $\tau$, Lemma \ref{tauoddeven} immediately yields

\begin{cor}
For any fixed consistent sequence $R^{\ast} = R(a_l, b_l)_{l \in [m]}$, 
$$\sum_{\substack{ {\mathbf P} \in S_n^{\on{good}}(\hat{\pi}_{\tau}),
\\ {\mathbf P} =\, n\Delta({\mathbf P}) + R^{\ast}}} \prod_{l=1}^m \erw(a_n(P_{1, l}) a_n(P_{2, g(l)})) = \begin{cases}
0, & m\ \text{\rm odd},\\
2^{m+1} \sigma^{2m}, & m\ \text{\rm even}.
\end{cases}
$$
\end{cor}

\begin{proof}[Proof of Lemma \ref{tauoddeven}] A matrix from the list \eqref{d3invar2} contributes a factor of $- \sigma^2$ if, and only if, its upper row is $(1, 1)$ or $(0, 0)$. So it suffices to consider only the upper row of the multi-index. By consistency, its first and its last entry must be equal, hence the number of steps, i.e., of occurrences of elements of the set $\{(1, 0), (0, 1)\}$ in the upper row must be even. So the number of occurrences of one of 
$(1, 1)$ or $(0, 0)$ is odd, if, and only if, $m$ is.
\end{proof}

\section{Class CI}
\label{c1}

In view of the similarities between the block structures of the CI and DIII classes, the description of the equivalence classes that was given in Lemma \ref{blocksd3} for the DIII case also applies to the CI case. In particular, we may reuse the terminology which was set up in Subsection \ref{d3fundamentals} and in particular the list of possible choices for the $\Delta_l$ in the present context. A crucial difference, though, concerns the way in which the various possibilities for $\Delta_l$ and $\Lambda_l$ translate into the sign of the expectation $\erw\left( a_n(p_{1l}, q_{1l}) a_n(p_{2l}, q_{2l}) \right)$ if one writes $ \left( \begin{array}{ll} p_{1l} & q_{1l}\\ p_{2l} & q_{2l} 
\end{array}\right) = n \Delta_l + \Lambda_l(a, b)$ with $a, b \in [n]$. Note that since the blocks of Class CI matrices are symmetric, the choice of $\Lambda_l$ does not affect the outcome. So we may fix $\Lambda_l = A(\cdot, \cdot)$ for all $l \in [m]$ as we study the role of $\Delta_l$. \\

If $\Delta_l$ is one of the matrices from \eqref{d3-delta-gleiche}, then $a_n(p_{2l}, q_{2l}) = a_n(p_{1l}, q_{1l})$, hence, the matrix entries being real, $\erw( a_n(p_{1l}, q_{1l}) a_n(p_{2l}, q_{2l}))
= \erw(|a_n(p_{1l}, q_{1l})|^2) = +\sigma^2.$\\

In the case $\Delta_l =  \left( \begin{array}{ll} 0 & 0\\ 1 & 1 
\end{array}\right)$, we have that $(p_{2l}, q_{2l}) = (n + p_{1l}, n+q_{1l})$ and thus $a_n(p_{2l}, q_{2l}) = - a_n(p_{1l}, q_{1l}).$ Therefore, $\erw( a_n(p_{1l}, q_{1l}) a_n(p_{2l}, q_{2l})) = - 
 \erw(|a_n(p_{1l}, q_{1l})|^2) = - \sigma^2$, and likewise for
 $\Delta_l =  \left( \begin{array}{ll} 1 & 1\\ 0&0 
\end{array}\right)$.\\

In the case $\Delta_l =  \left( \begin{array}{ll} 1 & 0\\ 0 & 1 
\end{array}\right)$, we have $(p_{1l}, q_{1l}) = (n + a_l, b_l)$ and $(p_{2l}, q_{2l}) = (a_l, n+b_l)$, hence $a_n(p_{1l}, q_{1l}) = 
a_n(p_{1l}, q_{1l})$. So we obtain that in this case, 
$\erw( a_n(p_{1l}, q_{1l}) a_n(p_{2l}, q_{2l})) = \sigma^2$, and likewise for $\Delta_l =  \left( \begin{array}{ll} 0 & 1 \\ 1 & 0 
\end{array}\right)$.\\

We now turn to the evaluation of 
\begin{equation}
\label{c1summeprodukt}
\sum_{\substack{ P \in S_n^{\on{good}}(\hat{\pi}_g),
\\ (\Delta_l({\mathbf P}), \Lambda_l({\mathbf P}))\\ \text{\rm substantial.}
}} \prod_{l=1}^m \erw(a_n(P_{1, l}) a_n(P_{2, g(l)}))
\end{equation} 
in the case that $g$ is a cyclic shift. As in Lemma \ref{gleichbleibtgleich} we have the dichotomy that either $(p_{1l}, q_{1l})= (p_{2l}, q_{2l})$ for all $l \in [m]$, or  
$(p_{1l}, q_{1l}) \neq (p_{2l}, q_{2l})$ for all $l \in [m]$. But now, for Class CI, the former case yields a contribution $+ \sigma^{2m}$, while in the latter case the contribution is
$$ (-1)^{m - \#\{ \text{\rm steps}\}} \sigma^{2m},$$
which is $(-1) \sigma^{2m}$ if $m$ is odd, and $\sigma^{2m}$ if $m$ is even. Consequently, for $m$ odd, the leading order contribution from the cyclic shifts vanishes.\\

Since for the case $m = 2d, d \in \nn$ the same combinatorial analysis as for the DIII class applies, the contribution of the 
cyclic shifts to the asymptotic variance is the same as in the case of the DIII class.\\

Recall from Subsection \ref{d3-tau} that in the case of the DIII class, the replacement of the domino condition (that arose in the analysis of cyclic shifts) by the reverse domino condition (that applied to the $\tau$ case) resulted in swapping the signs which are associated to each possible choice of $\Delta_l$. Now, observe that the swapped signs from Subsection \ref{d3-tau} are precisely the signs that in the case of the CI class apply to the analysis of cyclic shifts. On the other hand, we have already remarked that in view of the symmetries of the blocks of the CI matrices, replacing the domino by the reverse domino condition has no effect on the signs. Consequently, the analysis of the $\tau$ permutation for the DIII class, as spelled out in Subsection \ref{d3-tau}, applies verbatim to $\tau$ in the CI case. \\

In summary, the asymptotic variances in the case of the CI class are the same as those in the DIII case.


\begin{thebibliography}{10}

\bibitem{AK} G.~Abramovici, P.~Kalugin, \emph{Clifford modules and symmetries of topological insulators}, Int. J. Geom. Methods Mod. Phys. {\bf 09} (2012), 1250023.


\bibitem{AS}
A.~Altland and B.~Simons, \emph{Condensed matter field theory}, Cambridge UP,
  Cambridge, 2006.

\bibitem{AZ}
A.~Altland and M.~Zirnbauer, \emph{{Nonstandard symmetry classes in mesoscopic
  normal/superconducting hybrid structures}}, Physical Review B \textbf{55}
  (1997), no.~2, 1142--1161.




\bibitem{Dyson}
F.~J. Dyson, \emph{The threefold way. {A}lgebraic structure of symmetry groups
  and ensembles in quantum mechanics}, J. Mathematical Phys. \textbf{3} (1962),
  1199--1215. MR0177643 

\bibitem{azldp}
P.~Eichelsbacher and M.~Stolz, \emph{Large deviations for random matrix
  ensembles in mesoscopic physics}, Markov Process. Related Fields
\textbf{14} (2008), 207--232.




\bibitem{JMVA}
O.\ Friesen and M.\ L\"owe and M.\ Stolz,
\emph{Gaussian fluctuations for sample covariance matrices with
              dependent data}, J. Multivariate Anal. \textbf{114},(2013), 270--287
   



\bibitem{HHZ}
P.~Heinzner, A.~Huckleberry, and M.~R. Zirnbauer, \emph{Symmetry classes of
  disordered fermions}, Commun. Math. Phys. \textbf{257} (2005), 725--771. MR2164950

 
\bibitem{ECP} K.~Hofmann-Credner and M.~Stolz,
\emph{Wigner theorems for random matrices with
  dependent entries: Ensembles associated to symmetric spaces and sample
  covariance matrices}, Electronic Communications in Probability \textbf{13} (2008), 401--414. 


\bibitem{Isaacs} I.\ M.\ Isaacs, \emph{Finite Group Theory}, AMS 2008

\bibitem{Joh}  K.~Johansson,
\emph{On fluctuations of eigenvalues of random Hermitian matrices},
Duke Math. J. {\bf 91} (1998), 151-204.



\bibitem{KMS} T. Kusalik, J. A. Mingo and R. Speicher
\emph{Orthogonal polynomials and fluctuations of random matrices},
J.\ Reine Angew.\ Math.\ \textbf{604} (2007), 1--46.



\bibitem{NSp}
A.~Nica and R.~Speicher, \emph{Lectures on the combinatorics of free
  probability}, London Mathematical Society Lecture Note Series, vol. 335,
  Cambridge University Press, Cambridge, 2006. MR2266879


\bibitem{RSFL}  S.~Ryu, A.~P.~Schnyder, A.~Furusaki, A.~W.~W.~Ludwig, 
\emph{Topological insulators and superconductors: tenfold way and dimensional hierarchy},
New Journal of Physics {\bf 12} (2010), 065010-065069 .


\bibitem{SSB}
J.~H. Schenker and H.~Schulz-Baldes, \emph{Semicircle law and freeness for
  random matrices with symmetries or correlations}, Math. Res. Lett.
  \textbf{12} (2005), no.~4, 531--542. MR2155229 

\bibitem{SSB2}
J.~H. Schenker and H.~Schulz-Baldes, \emph{Gaussian fluctuations for random matrices 
with correlated entries}, Int. Math. Res. Not. 2007, Article ID rnm047, doi:10.1093/imrn/rnm047.
MR2348645

\bibitem{Wigner58}
E.~P. Wigner, \emph{On the distribution of the roots of certain symmetric
  matrices}, Ann. of Math. (2) \textbf{67} (1958), 325--327. MR0095527

\bibitem{ZirnHbk}
M.\ R.\ Zirnbauer, \emph{Symmetry Classes}, in: G.\ Akemann et al.\ (eds.), The Oxford Handbook of Random Matrix Theory, Oxford UP 2011, 43--65.

\end{thebibliography}
\end{document}